\documentclass[11pt]{article}
\usepackage{amssymb}
\usepackage{colortbl}
\usepackage{amsfonts,amsmath, longtable}

        \def\theequation{\thesection.\arabic{equation}}

\newcommand{\ti}[1]{\tilde{#1}}

\newcommand{\mF}{{\mathcal F}}
\newcommand{\mR}{{\mathcal R}}
\newcommand{\mH}{{\mathcal H}}

\newcommand{\vf}{\varphi}
\newcommand{\al}{\alpha}
\newcommand{\be}{\beta}

\newcommand{\om}{\omega}
\newcommand{\vth}{\vartheta}

\newcommand{\MatM}{ {\rm Mat}(M,\mathbb C) }

\newcommand{\mC}{\mathbb C}
\newcommand{\mZ}{\mathbb Z}

\newcommand{\ka}{\kappa}

\newtheorem{lemma}{Lemma}[section]

\newenvironment{proof}{\par\noindent{\bf Proof.}}{\hfill$\scriptstyle\blacksquare$}

\def\beq{\begin{equation}}
\def\eq{\end{equation}}

\def\res{\mathop{\hbox{Res}}\limits}

\begin{document}

\setcounter{page}{1}

\begin{center}

\

\vspace{-0mm}

{\Large{\bf On R-matrix identities related to elliptic  }}

\vspace{3mm}

{\Large{\bf anisotropic spin Ruijsenaars-Macdonald operators}}

 \vspace{15mm}

 {\Large {M. Matushko}}$\,^{\diamond\,*}$
\qquad\quad\quad
 {\Large {A. Zotov}}$\,^{\diamond\,\bullet}$

  \vspace{5mm}

$\diamond$ -- {\em Steklov Mathematical Institute of Russian
Academy of Sciences,\\ Gubkina str. 8, 119991, Moscow, Russia}

$*$ -- {\em Center for Advanced Studies, Skoltech, 143026, Moscow, Russia}

$\bullet$ -- {\em Institute for Theoretical and Mathematical Physics,\\ Lomonosov Moscow State University, Moscow, 119991, Russia}

   \vspace{3mm}

 {\small\rm {e-mails: matushko@mi-ras.ru, zotov@mi-ras.ru}}

\end{center}

\vspace{0mm}

\begin{abstract}
We propose and prove a set of identities for ${\rm GL}_M$ elliptic $R$-matrix (in the fundamental representation).
 In the scalar case ($M=1$) these are elliptic function identities derived by S.N.M. Ruijsenaars as necessary and sufficient conditions for his kernel identity underlying
construction of integral solutions to quantum spinless Ruijsenaars-Schneider model.
In this respect the result
of the present paper can be considered as the first step towards constructing solutions of quantum eigenvalue problem for
the anisotropic spin Ruijsenaars model.
\end{abstract}

%

{\small{
\tableofcontents
}}


\section{Introduction}
\setcounter{equation}{0}

In his paper \cite{Ruij} S.N.M. Ruijsenaars introduced a set of commuting elliptic difference operators
\begin{equation}\label{Dscalar}
 D_k=\sum\limits_{\substack{|I|=k}}\prod\limits_{\substack{i\in I \\ j\notin I}}\phi(z_j-z_i)\prod_{i\in I}p_{i},\qquad k=1,\dots,N\,,
\end{equation}
where $\phi(z)=\phi(\hbar,z)$ is the elliptic Kronecker function (\ref{a0962}).
The sum in (\ref{Dscalar}) is over all subsets $I$ of size $k$ of $\{1,\dots,N\}={\mathcal N}$.
The shift operators $p_i$ act on a function $f(z_1,\dots ,z_N)$ of complex variables ${\bf z}=z_1,\dots ,z_N$ as follows:
\beq\label{p_i}
(p_if)(z_1,z_2,\dots z_N)=\exp\left(-\eta \frac{\partial}{\partial z_i}\right)f(z_1,\dots,z_N)=f(z_1,\dots,z_i-\eta,\dots, z_N).
\eq
 The operators (\ref{Dscalar}) are called
the Ruijsenaars-Macdonald operators since in the trigonometric limit they reproduce the
Macdonald operators.

The proof of mutual commutativity of $D_k$ was shown in \cite{Ruij} to be equivalent to the following set
of functional equations:
\begin{equation}\label{IdenRuij}
\sum_{|I|=k}\left(\prod\limits_{\substack{i\in I\\ j\notin I}}\phi(z_j-z_i)\phi(z_i-z_j-\eta)-\prod\limits_{\substack{i\in I\\ j\notin I}}\phi(z_i-z_j)\phi(z_j-z_i-\eta)\right)=0\,,\quad
k=1,\dots,N\,,
\end{equation}
and the function $\phi(z)$  (\ref{a0962}) was shown to satisfy (\ref{IdenRuij}).

Next, it was shown in \cite{Ruij2006} that the kernel identity $(D_k({\bf x})-D_k(-{\bf y}))\Psi(\bf x,\bf y)=0$
holds true for $\Psi(\bf x,\bf y)$ given by a certain product (and/or ratio) of elliptic gamma functions. It allows to construct
integral solutions to quantum eigenproblem for operators $D_k$ \cite{HR,Komori,KK}. This approach was later extended
to a more general class of operators \cite{Lang}. The proof of the kernel identity is performed by a direct substitution which provides a set of identities for the function $\phi$. Namely, it was proved in \cite{Ruij2006} that the kernel identities is valid iff the following set of relations for $2N$ complex variables (${\bf x}=x_1,...,x_N$ and ${\bf y}=y_1,...,y_N$) holds:
 \begin{equation}
 \label{e01}
\displaystyle{
\sum\limits_{\substack{I\subset {\mathcal N} \\|I|=k}} \left(
\prod\limits_{\substack{i\in I \\ j\in I^c}} \phi(x_i-x_j)
\prod\limits_{\substack{i\in I \\ j\in {\mathcal N}}} \phi(y_j-x_i)
-
\prod\limits_{\substack{i\in I \\ j\in I^c}} \phi(y_j-y_i)
\prod\limits_{\substack{i\in I \\ j\in {\mathcal N}}} \phi(y_i-x_j)
\right)
=0\,,
\quad
k=1,\ldots,N\,,
  }
\end{equation}
where  $I^c$ means the complement of a set $I$ in $\mathcal N$. It was shown in \cite{Ruij2006} that the the elliptic Kronecker function (\ref{a0962}) satisfies (\ref{e01}). For example, for $k=1$ (\ref{e01}) yields
 \begin{equation}
 \label{e02}
\displaystyle{
\sum\limits_{i=1}^N \left(
\prod\limits_{j\neq i}^N \phi(x_i-x_j)
\prod\limits_{j=1}^N \phi(y_j-x_i)
-
\prod\limits_{j\neq i}^N \phi(y_j-y_i)
\prod\limits_{j=1}^N \phi(y_i-x_j)
\right)
=0\,.
  }
\end{equation}
Under the substitution $x_i=z_i$, $y_i=z_i-\eta$ the latter identity reproduces the first one (with $k=1$) from (\ref{IdenRuij}). For $k>1$ the substitution $x_i=z_i$, $y_i=z_i-\eta$ provides particular case of (\ref{IdenRuij}) corresponding to  $\hbar=2\eta$.

Let us mention that the identities (\ref{e01}) arise in the context of B\"acklund transformations in the
classical (elliptic) Ruijsenaars-Schneider model \cite{HR}. Also, (\ref{e02}) provides the so-called self-dual form of the Ruijsenaars-Schneider model coming from the multi-pole ansatz for intermediate long-wave equation \cite{ZZ2}.

\paragraph{Purpose of the paper} is to extend the identities (\ref{e01}) to $R$-matrix identities in the following sense. The elliptic  Baxter-Belavin \cite{Baxter,Belavin} $R$-matrix (\ref{BB}) in the fundamental representation of ${\rm GL}_M$ Lie group can be considered as matrix generalization of the Kronecker function (\ref{a0962}) because
(\ref{a0962}) is a particular case of (\ref{BB}) corresponding to $M=1$. $R$-matrices do not commute (in the general case). Our aim is to find noncommutative analogue of (\ref{e01}) which reproduce (\ref{e01}) in the $M=1$ case.

In our previous paper \cite{MZ1} (see also \cite{MZ2} for applications to long-range spin chains) we introduced anisotropic spin generalization for the operators (\ref{Dscalar}). Similarly to the scalar case \cite{Ruij} we deduced and proved $R$-matrix generalizations of the
identities (\ref{IdenRuij}). These $R$-matrix identities are necessary and sufficient conditions for commutativity of spin Ruijsenaars-Macdonald operators. In this paper we propose and prove $R$-matrix generalizations for the identities
(\ref{e01}) from \cite{Ruij2006}:
\beq\label{e021}
  \begin{array}{c}
  \displaystyle{
   \sum\limits_{1\leq i_1<...<i_k\leq N}\left(
   \overrightarrow{\prod\limits_{l_k=i_k+1}^N} R_{i_k l_k}^{\hbar}
   \overrightarrow{\prod\limits^N_{\hbox{\tiny{$ \begin{array}{c}{ l_{k-1}\!=\!i_{k-1}\!+\!1 }\\{ l_{k-1}\!\neq\! i_k } \end{array}$}}}}R_{i_{k-1} l_{k-1}}^{\hbar}\right.
      \ \ldots\
   \overrightarrow{\prod\limits^N_{\hbox{\tiny{$ \begin{array}{c}{ l_{1}\!=\!i_{1}\!+\!1 }\\{ l_{1}\!\neq\! i_2...i_k } \end{array}$}}}}R_{i_{1} l_{1}}^{\hbar}
   \times
   }
   \end{array}
 \eq
   $$
   \begin{array}{c}
     \displaystyle{
 \times
 \overrightarrow{\prod\limits^{2N}_{j_1=N+1}} R_{j_1i_1}^{\hbar}
 \overrightarrow{\prod\limits^{2N}_{j_2=N+1}} R_{j_2i_2}^{\hbar}
 \ \ldots\
 \overrightarrow{\prod\limits^{2N}_{j_k=N+1}} R_{j_ki_k}^{\hbar}
  \times
    }
   \\ \ \\
     \displaystyle{
       \times \left.
   \overrightarrow{\prod\limits^{i_k-1}_{\hbox{\tiny{$ \begin{array}{c}{ m_k\!=\!1 }\\{ m_k\!\neq\! i_{1}...i_{k-1}} \end{array}$}}}}R_{i_k m_k}^{\hbar}
   \overrightarrow{\prod\limits^{i_{k-1}-1}_{\hbox{\tiny{$ \begin{array}{c}{ m_{k-1}\!=\!1 }\\{ m_{k-1}\!\neq\! i_{1}...i_{k-2}} \end{array}$}}}}R_{i_{k-1} m_{k-1}}^{\hbar}
      \ \ldots\
  \overrightarrow{\prod\limits^{i_{1}-1}_ {m_1=1}} R_{i_{1} m_{1}}^{\hbar}\right)-
 }
 \end{array}
 $$
 $$
  \begin{array}{c}
  \displaystyle{
   - \sum\limits_{N+1\leq j_1<...<j_k\leq 2N}\left(
   \overrightarrow{\prod\limits_{l_k=j_k+1}^{2N}} R_{l_k j_k }^{\hbar}
   \overrightarrow{\prod\limits^{2N}_{\hbox{\tiny{$ \begin{array}{c}{ l_{k-1}\!=\!j_{k-1}\!+\!1 }\\{ l_{k-1}\!\neq\! j_k } \end{array}$}}}}R_{l_{k-1}j_{k-1} }^{\hbar}\right.
      \ \ldots\
   \overrightarrow{\prod\limits^{2N}_{\hbox{\tiny{$ \begin{array}{c}{ l_{1}\!=\!j_{1}\!+\!1 }\\{ l_{1}\!\neq\! j_2...j_k } \end{array}$}}}}R_{l_{1}j_{1} }^{\hbar}
   \times
   }
   \end{array}
 $$
   $$
   \begin{array}{c}
     \displaystyle{
 \times
 \overrightarrow{\prod\limits^{N}_{i_1=1}} R_{j_1i_1}^{\hbar}
 \overrightarrow{\prod\limits^{N}_{i_2=1}} R_{j_2i_2}^{\hbar}
 \ \ldots\
 \overrightarrow{\prod\limits^{N}_{i_k=1}} R_{j_ki_k}^{\hbar}
  \times
    }
   \\ \ \\
     \displaystyle{
       \times \left.
   \overrightarrow{\prod\limits^{j_k-1}_{\hbox{\tiny{$ \begin{array}{c}{ m_k\!=\!N+1 }\\{ m_k\!\neq\! j_{1}...j_{k-1}} \end{array}$}}}}R_{ m_k j_k}^{\hbar}
   \overrightarrow{\prod\limits^{j_{k-1}-1}_{\hbox{\tiny{$ \begin{array}{c}{ m_{k-1}\!=\!N+1 }\\{ m_{k-1}\!\neq\! j_{1}...j_{k-2}} \end{array}$}}}}R_{m_{k-1}j_{k-1} }^{\hbar}
      \ \ldots\
  \overrightarrow{\prod\limits^{j_{1}-1}_ {m_1=N+1}} R_{ m_{1}j_{1}}^{\hbar}\right)=0\,.
 }
 \end{array}
 $$
Notations are explained in the next Section.
 This result is the first step towards constructing the kernel identity and solutions to quantum problem of the (anisotropic) spin Ruijsenaars model. The latter is quite nontrivial problem since it requires a definition of matrix generalization of elliptic gamma function. We hope to clarify these questions in our future publications.

\section{Spin operators and R-matrix identities}
\setcounter{equation}{0}

Here we recall some notations and statements from \cite{MZ1}. Some additional notations are introduced
for the purposes of the present paper.

\paragraph{$R$-matrices and Yang-Baxter equations.}
Let ${\mathcal H}$ be a vector space ${\mathcal H}=(\mC^M)^{\otimes N}$.
The elliptic Baxter-Belavin $R$-matrix \cite{Baxter,Belavin}  is a linear map
(\ref{BB})
$R^\hbar_{ij}(z)\in {\rm End}({\mathcal H})$
acting non-trivially in the $i$-th and $j$-th tensor components of ${\mathcal H}$ only and
satisfying the quantum Yang-Baxter equation:
\beq\label{QYB2}
\begin{array}{c}
\displaystyle{
    R^{\hbar}_{ij}(u)  R^{\hbar}_{ik}(u+v) R^{\hbar}_{jk}(v) =
      R^{\hbar}_{jk}(v) R^{\hbar}_{ik}(u+v) R^{\hbar}_{ij}(u)
      }
\end{array}\eq
for any distinct integers $1\leq i,j,k\leq N$. Also, for any distinct integers $1\leq i,j,k,l\leq N$
\beq\label{QYB3}
\begin{array}{c}
\displaystyle{
    [R^{\hbar}_{ij}(u), R^{\hbar'}_{kl}(v)]=0\,.
      }
\end{array}\eq
Consider a pair $I,J$ of disjoint subsets in $\{1,\dots,N\}$. The elements in $I,J$ are in increasing ordering, that is for $J=\{j_1,j_2,\dots,j_k\}$ the elements $j_m$ are in increasing order $j_1<j_2<\dots<j_k$.
Similarly, for $I=\{i_1,i_2,\dots,i_l\}$ we have $i_1<i_2<\dots<i_l$.
We will use notations
$\overrightarrow{\prod\limits^{N}_ {j=1}} R_{ij}$ and $\overleftarrow{\prod\limits^{N}_ {j=1}} R_{ij}$, where the arrows mean the ordering. For example,
$\overrightarrow{\prod\limits^{N}_ {j=1}} R_{ij}=R_{i1}R_{i2}...R_{iN}$ and $\overleftarrow{\prod\limits^{N}_ {j=1}} R_{ij}=R_{iN}R_{i,N-1}...R_{i1}$.

Introduce the following notation:
\begin{equation}\label{RIJ1}
\mathcal{R}_{I,J}=\overleftarrow{\prod\limits_{\substack{i_1\in I\\i_1<j_1}}} R_{i_1,j_1}\left({z_{i_1}}-{z_{j_1}}\right) \overleftarrow{\prod\limits_{\substack{i_2\in I\\i_2<j_2}}} R_{i_2,j_2}\left({z_{i_2}}-{z_{j_2}}\right)\dots\overleftarrow{\prod\limits_{\substack{i_k\in I\\i_k<j_k}}}R_{i_k,j_k}\left({z_{i_k}}-{z_{j_k}}\right)\,.
\end{equation}
Using the property (\ref{QYB3}) it is easy to rewrite it as
\begin{equation}\label{RIJ2}
\mathcal{R}_{I,J}=\overrightarrow{\prod\limits_{\substack{j_l\in J\\j_l>i_l}}}R_{i_l,j_l}\left({z_{i_l}}-{z_{j_l}}\right) \overrightarrow{\prod\limits_{\substack{j_{l-1}\in J\\ j_{l-1}>i_{l-1}}}} R_{i_{l-1},j_{l-1}}\left({z_{i_{l-1}}}-{z_{j_{l-1}}}\right)\dots\overrightarrow{\prod\limits_{\substack{j_1\in J\\j_1>i_1}}} R_{i_1,j_1}\left({z_{i_1}}-{z_{j_1}}\right)\,.
\end{equation}
In a similar way we also introduce
\begin{equation}\label{RIJ'1}
\mathcal{R}_{I,J}'=\overrightarrow{\prod\limits_{\substack{j_l\in J\\j_l<i_l}}}R_{i_l,j_l}(z_{i_l}-z_{j_l})\overrightarrow{\prod\limits_{\substack{j_{l-1}\in J\\j_{l-1}<i_{l-1}}}} R_{i_{l-1},j_{l-1}}(z_{i_{l-1}}-z_{j_{l-1}})\dots \overrightarrow{\prod\limits_{\substack{j_1\in J\\j_1<i_1}}} R_{i_1,j_1}(z_{i_1}-z_{j_1})=
\end{equation}

\begin{equation}\label{RIJ'2}
=\overleftarrow{\prod\limits_{\substack{i_1\in I\\i_1>j_1}}} R_{i_1,j_1}\left({z_{i_1}}-{z_{j_1}}\right) \overleftarrow{\prod\limits_{\substack{i_2\in I\\i_2>j_2}}} R_{i_2,j_2}\left({z_{i_2}}-{z_{j_2}}\right)\dots\overleftarrow{\prod\limits_{\substack{i_k\in I\\i_k>j_k}}} R_{i_k,j_k}\left({z_{i_k}}-{z_{j_k}}\right)\,.
\end{equation}
For any disjoint subsets $A,B,C$ of $\{1,2,\dots,N\}$ the following identities hold true:
\begin{equation}\label{l21}
 \mathcal{R}_{C,A\cup B}\mathcal{R}_{B,A}=\mathcal{R}_{B\cup C,A}\mathcal{R}_{C,B}\,,
\end{equation}
\begin{equation}\label{l22}
 \mathcal{R}'_{A,B}\mathcal{R}'_{A\cup B,C}=\mathcal{R}'_{B,C}\mathcal{R}'_{A,B\cup C}\,.
\end{equation}
The latter identities follow from the Yang-Baxter equation (\ref{QYB2}).

Elliptic $R$-matrix satisfies the unitarity property
\beq\label{q03}\begin{array}{c}
    R^{\hbar}_{ij}(z) R^\hbar_{ji}(-z)= {\rm Id}\, \phi(\hbar,z)\phi(\hbar,-z)\,,
\end{array}\eq
Below we also use $R$-matrices in a slightly different normalization:
\beq\label{q04}
\begin{array}{c}
    R^{\hbar}_{ij}(z) = \phi(\hbar,z){\bar R}^{\hbar}_{ij}(z)\,.
\end{array}\eq
Then
\beq\label{q05}
\begin{array}{c}
    {\bar R}^{\hbar}_{ij}(z) {\bar R}^\hbar_{ji}(-z)= {\rm Id}\,.
\end{array}\eq
For the products defined above we have:
\begin{equation}\label{un01}
 \mathcal{R}_{I,J}\mathcal{R}_{J,I}'={\rm Id}\prod\limits_{\substack{i\in I,j\in J\\i<j}}\phi(h,z_i-z_j)\phi(h,z_j-z_i)
\end{equation}
and
\begin{equation}\label{un02}
 \mathcal{R}'_{I,J}\mathcal{R}_{J,I}={\rm Id}\prod\limits_{\substack{i\in I,j\in J\\i>j}}\phi(h,z_i-z_j)\phi(h,z_j-z_i)\,.
\end{equation}
Similarly,
\begin{equation}\label{un03}
 \bar{\mathcal{R}}_{I,J}\bar{\mathcal{R}}_{J,I}'=\bar{\mathcal{R}}_{J,I}'\bar{\mathcal{R}}_{I,J}={\rm Id}\,.
\end{equation}
Finally, the elliptic $R$-matrix (\ref{BB}) is skew-symmetric:
 \beq\label{r08}
 \begin{array}{c}
  \displaystyle{
 R^{\hbar}_{12}(z)=-R^{-\hbar}_{21}(-z)\,.
  }
 \end{array}
 \eq
 Below by writing $R^{-\hbar}_{12}(z)$ we assume $R^{\hbar}_{21}(-z)$.

In this paper we also use the following notations. For disjoint sets $I=\{i_1<i_2<\dots <i_k\}$ and $J=\{j_1<j_2<\dots<j_l\}$ we denote by ${\mathcal Y}_{I,J}^{\hbar} $ the following product of $R$-matrices
  \beq\label{Y1}
 {\mathcal Y}^{\hbar} _{I,J}=\overrightarrow{\prod_{j\in J}}R_{i_1,j}^{\hbar}  \dots\overrightarrow{\prod_{j\in J}}R_{i_k,j}^{\hbar}\,.
 \eq
 Due to (\ref{QYB3}) it can be written as
 \beq\label{Y2}
 {\mathcal Y}^{\hbar} _{I,J}=\overrightarrow{\prod_{i\in I}}R_{i,j_1}^{\hbar}  \dots\overrightarrow{\prod_{i\in I}}R_{i,j_l}^{\hbar}\,.
 \eq
  In what follows we use the auxiliary statement based on the Yang-Baxter equation (\ref{QYB2})-(\ref{QYB3}).
  \begin{lemma}\label{lemmaY}
  Let $I=\{i_1<\dots<i_{m-1}<i_m=a<i_{m+1}<\dots <i_k\}$ and $J=\{j_1<\dots<j_{n-1}<j_n=b<j_{n+1}<\dots<j_l\}$ be disjoint sets.
  Then the following relations hold true:
  \beq\label{lemmaY1}
  {\mathcal Y}^{\hbar}_{I,J}R^{\hbar}_{b,j_{n+1}}\dots R^{\hbar}_{b,j_l}=R^{\hbar}_{b,j_{n+1}}\dots R^{\hbar}_{b,j_l}{\mathcal Y}^{\hbar}_{I,J\setminus \{b\}}{\mathcal Y}^{\hbar}_{I,\{b\}}
  \eq
  and
  \beq\label{lemmaY2}
 {\mathcal Y}^{\hbar}_{I,J}R^{\hbar}_{i_{m+1},a}\dots R^{\hbar}_{i_k,a}
 =R^{\hbar}_{i_{m+1},a}\dots R^{\hbar}_{i_k,a}{\mathcal Y}^{\hbar}_{I\setminus \{a\},J}{\mathcal Y}^{\hbar}_{\{a\},J}
 \,.
  \eq
  \end{lemma}
  The proof is given in the Appendix.

\paragraph{Spin operators} are defined as follows:
\begin{equation}
\label{Dspin2}
\mathcal{D}_k=
\sum_{|I|=k}\ \prod\limits_{\substack{i\in I^c,\\ j\in I}}\phi(z_{i}-z_j) \cdot\bar{\mathcal{R}}_{I^c,I}\cdot \mathbf{p}_{I}\cdot \bar{\mathcal{R}}_{I,I^c}'\,,\qquad \mathbf{p}_I=\prod_{i\in I}p_{i}\,,
\end{equation}
where $p_i$ are the shift operators (\ref{p_i}) and the bars over $\mR$ mean that in the definitions (\ref{RIJ1}), (\ref{RIJ'1})
we use $R$-matrices normalized as in (\ref{q04})-(\ref{q05}). For $M=1$ the operators (\ref{Dspin2}) reproduce
the Ruijsenaars-Macdonald operators (\ref{Dscalar}). The operators of type (\ref{Dspin2}) were introduced
in \cite{LPS} in trigonometric case with $M=2$ based on the ${\rm U}_q({\rm gl}_2)$ $R$-matrix.

Let us also remark that using the so-called freezing trick
the spin operators (\ref{Dspin2}) provide elliptic integrable long-range spin chains \cite{MZ2}.
In the classical mechanics the operators (\ref{Dscalar})
are the Hamiltonians of the elliptic Ruijsenaars-Schneider model, while the (anisotropic) spin operators are quantum analogues of the Hamiltonians of (relativistic) interacting tops \cite{Z19,GSZ}. Description of relativistic tops can be found in \cite{LOZ14,KZ19}.

Commutativity of $\mathcal{D}_k$ turns out to be equivalent to the following set of $R$-matrix identities:
\begin{equation}\label{relRI}
\sum_{|I|=k}\left(\mathcal{R}_{I^c,I}\cdot\mathcal{R}_{I_-,I^c}'
\cdot\mathcal{R}_{I_-,I^c}\cdot\mathcal{R}_{I^c,I}'-\mathcal{R}_{I,I^c}
\cdot\mathcal{R}_{I^c_-,I}'\cdot\mathcal{R}_{I^c_-,I}\cdot\mathcal{R}_{I,I^c}'\right)=0\,,
\end{equation}
which hold true for any $k=1,\dots ,N$. Here we use the notation $\mathcal{R}_{I_-,J}=\mathbf{p}_I \mathcal{R}_{I,J}\mathbf{p}_I^{-1}$.  Derivation of (\ref{relRI}) as well as its proof and explicit examples can be found in \cite{MZ1}. For $M=1$ (\ref{relRI}) turn into identities (\ref{IdenRuij}).

\section{New R-matrix identities}
\setcounter{equation}{0}

\subsection{Statement}
It is convenient to define a set of $2N$ complex variables $z_1,...,z_{2N}$, which unifies $\bf x$ and $\bf y$:
 \begin{equation}
 \label{e51}
 \begin{array}{c}
z_i=
 \left\{\begin{array}{l}
\displaystyle{
x_i\quad \hbox{for}\ i\in{\mathcal N}_1\,,\quad {\mathcal N}_1=\{1,\ldots,N\}\,,
  }
  \\ \ \\
 \displaystyle{
y_i\quad \hbox{for}\ i\in{\mathcal N}_2\,,\quad {\mathcal N}_2=\{N+1,\ldots,2N\}\,,
  }
  \end{array}
  \right.
  \end{array}
\end{equation}
The $R$-matrices act on the space $\mH=(\mC^M)^{\otimes 2N}$, and the following short notations are assumed:
 \begin{equation}
 \label{e52}
 \begin{array}{c}
 R^\hbar_{ij}=
 \left\{\begin{array}{l}
\displaystyle{
R^\hbar_{ij}(x_i-x_j)\,,\quad \hbox{if}\ i,j\in {\mathcal N}_1\,,
  }
  \\ \ \\
 \displaystyle{
R^\hbar_{ij}(x_i-y_j)\,,\quad \hbox{if}\ i\in {\mathcal N}_1\,,\ j\in {\mathcal N}_2\,,
  }
    \\ \ \\
 \displaystyle{
R^\hbar_{ij}(y_i-x_j)\,,\quad \hbox{if}\ i\in {\mathcal N}_2\,,\ j\in {\mathcal N}_1\,,
  }
   \\ \ \\
 \displaystyle{
R^\hbar_{ij}(y_i-y_j)\,,\quad \hbox{if}\ i,j\in {\mathcal N}_2\,.
  }
    \end{array}
  \right.
  \end{array}
\end{equation}

\paragraph{Theorem}
{\em
The following identities hold for any $k=1,\dots ,N$:
\begin{equation}
\label{e53}
\begin{array}{lll}
 \displaystyle \sum\limits_{\substack{I\subset {\mathcal N}_1\\|I|=k}} \mathcal{R}_{I,I^c}^{h}{\mathcal Y}^{\hbar}_{\mathcal N_2,I}\mathcal{R'}_{I,I^c}^{\hbar}+
 (-1)^{k-1}\displaystyle \sum\limits_{\substack{J\subset {\mathcal N}_2\\|J|=k}} \mathcal{R}_{J,J^c}^{-\hbar} {\mathcal Y}^{-\hbar}_{\mathcal N_1,J}\mathcal{  R'}_{J,J^c}^{-\hbar}=0\,,
\end{array}
 \end{equation}
 where $R_{ij}=R^\hbar_{ij}(z_i-z_j)$ is the elliptic $R$-matrix (\ref{BB}) acting in the $i$-th and $j$-th tensor components, $i,j=1,...,2N$ according to the rule\footnote{In the first sum of (\ref{e53}) the set $I^c$ is a complement of $I$ in ${\mathcal N}_1$. Similarly, in the second sum of (\ref{e53}) the set $J^c$ is a complement of $J$ in ${\mathcal N}_2$.} (\ref{e51})-(\ref{e52}).
 In the $M=1$ case the $R$-matrix identities (\ref{e53}) turn into (\ref{e01}).
 }

 The idea of the proof is as follows. Denote the l.h.s. of (\ref{e53}) as $\mF[k,N]=\mF_1[k,N]-\mF_2[k,N]$, where
 $\mF_1[k,N],\mF_2[k,N]$ are the first and the second sum in (\ref{e53}) respectively\footnote{Sometimes we omit the indices $k, N$ and write $\mF$ instead of $\mF[k,N]$.} :
  \begin{equation}
\label{e531}
\begin{array}{lll}
 \displaystyle
  \mF_1[k,N]=\sum\limits_{\substack{I\subset {\mathcal N}_1\\|I|=k}} \mathcal{R}_{I,I^c}^{h}{\mathcal Y}^{\hbar}_{\mathcal N_2,I} \mathcal{R'}_{I,I^c}^{\hbar}
  \,,
\end{array}
 \end{equation}
  \begin{equation}
\label{e532}
\begin{array}{lll}
 \displaystyle
 \mF_2[k,N]=(-1)^k\sum\limits_{\substack{J\subset {\mathcal N}_2\\|J|=k}} \mathcal{R}_{J,J^c}^{-\hbar} {\mathcal Y}^{-\hbar}_{\mathcal N_1,J} \mathcal{  R'}_{J,J^c}^{-\hbar}
 \,.
\end{array}
 \end{equation}
 More explicitly,
 \beq\label{e534}
  \begin{array}{c}
  \displaystyle{
   {\mathcal F}_{1}[k,N]=  \sum\limits_{1\leq i_1<...<i_k\leq N}\left(
   \overrightarrow{\prod\limits_{l_k=i_k+1}^N} R_{i_k l_k}^{\hbar}
   \overrightarrow{\prod\limits^N_{\hbox{\tiny{$ \begin{array}{c}{ l_{k-1}\!=\!i_{k-1}\!+\!1 }\\{ l_{k-1}\!\neq\! i_k } \end{array}$}}}}R_{i_{k-1} l_{k-1}}^{\hbar}\right.
      \ \ldots\
   \overrightarrow{\prod\limits^N_{\hbox{\tiny{$ \begin{array}{c}{ l_{1}\!=\!i_{1}\!+\!1 }\\{ l_{1}\!\neq\! i_2...i_k } \end{array}$}}}}R_{i_{1} l_{1}}^{\hbar}
   \times
   }
   \end{array}
 \eq
   $$
   \begin{array}{c}
     \displaystyle{
 \times
 \overrightarrow{\prod\limits^{2N}_{j_1=N+1}} R_{j_1i_1}^{\hbar}
 \overrightarrow{\prod\limits^{2N}_{j_2=N+1}} R_{j_2i_2}^{\hbar}
 \ \ldots\
 \overrightarrow{\prod\limits^{2N}_{j_k=N+1}} R_{j_ki_k}^{\hbar}
  \times
    }
   \\ \ \\
     \displaystyle{
       \times \left.
   \overrightarrow{\prod\limits^{i_k-1}_{\hbox{\tiny{$ \begin{array}{c}{ m_k\!=\!1 }\\{ m_k\!\neq\! i_{1}...i_{k-1}} \end{array}$}}}}R_{i_k m_k}^{\hbar}
   \overrightarrow{\prod\limits^{i_{k-1}-1}_{\hbox{\tiny{$ \begin{array}{c}{ m_{k-1}\!=\!1 }\\{ m_{k-1}\!\neq\! i_{1}...i_{k-2}} \end{array}$}}}}R_{i_{k-1} m_{k-1}}^{\hbar}
      \ \ldots\
  \overrightarrow{\prod\limits^{i_{1}-1}_ {m_1=1}} R_{i_{1} m_{1}}^{\hbar}\right)\,.
 }
 \end{array}
 $$
 Let us also rewrite expression (\ref{e532}) through $R$-matrices depending on $\hbar$ (but not on $-\hbar$) using the skew-symmetry property (\ref{r08}):
 \beq\label{e535}
  \begin{array}{c}
  \displaystyle{
   {\mathcal F}_{2}[k,N]= \sum\limits_{N+1\leq j_1<...<j_k\leq 2N}\left(
   \overrightarrow{\prod\limits_{l_k=j_k+1}^{2N}} R_{l_k j_k }^{\hbar}
   \overrightarrow{\prod\limits^{2N}_{\hbox{\tiny{$ \begin{array}{c}{ l_{k-1}\!=\!j_{k-1}\!+\!1 }\\{ l_{k-1}\!\neq\! j_k } \end{array}$}}}}R_{l_{k-1}j_{k-1} }^{\hbar}\right.
      \ \ldots\
   \overrightarrow{\prod\limits^{2N}_{\hbox{\tiny{$ \begin{array}{c}{ l_{1}\!=\!j_{1}\!+\!1 }\\{ l_{1}\!\neq\! j_2...j_k } \end{array}$}}}}R_{l_{1}j_{1} }^{\hbar}
   \times
   }
   \end{array}
 \eq
   $$
   \begin{array}{c}
     \displaystyle{
 \times
 \overrightarrow{\prod\limits^{N}_{i_1=1}} R_{j_1i_1}^{\hbar}
 \overrightarrow{\prod\limits^{N}_{i_2=1}} R_{j_2i_2}^{\hbar}
 \ \ldots\
 \overrightarrow{\prod\limits^{N}_{i_k=1}} R_{j_ki_k}^{\hbar}
  \times
    }
   \\ \ \\
     \displaystyle{
       \times \left.
   \overrightarrow{\prod\limits^{j_k-1}_{\hbox{\tiny{$ \begin{array}{c}{ m_k\!=\!N+1 }\\{ m_k\!\neq\! j_{1}...j_{k-1}} \end{array}$}}}}R_{ m_k j_k}^{\hbar}
   \overrightarrow{\prod\limits^{j_{k-1}-1}_{\hbox{\tiny{$ \begin{array}{c}{ m_{k-1}\!=\!N+1 }\\{ m_{k-1}\!\neq\! j_{1}...j_{k-2}} \end{array}$}}}}R_{m_{k-1}j_{k-1} }^{\hbar}
      \ \ldots\
  \overrightarrow{\prod\limits^{j_{1}-1}_ {m_1=N+1}} R_{ m_{1}j_{1}}^{\hbar}\right)\,.
 }
 \end{array}
 $$
 In this way (\ref{e53}) takes the form (\ref{e021}).

  Consider $\mF$ as (matrix valued) function of $z_1=x_1$. Our strategy is to use Lemma \ref{Lem41} from the Appendix similarly to the proof of (\ref{e01}) in \cite{Ruij2006}. We are going to show that $\mF$ is an entire function of $z_1$. The behaviour of $R$-matrices on the lattice of periods $\mZ\oplus\tau\mZ$ is nontrivial. It is given by (\ref{r721}). For this reason we use the approach from \cite{MZ1}. Namely, we consider $\mF(z_1)$
 on the ''large'' torus with periods $M$ and $M\tau$. Due to (\ref{a041}) $Q^M=\Lambda^M=1_M$ the quasi-periodic behaviour of $R$-matrices (and $\mF$) on the large torus is simplified and we may use Lemma \ref{Lem41}.
 First, we are going to prove the absence of poles (in the variable $z_1$) in the fundamental parallelogram
 $0\leq z_1<1$, $0\leq z_1<\tau$ and then extend the proof to the large torus.

The proof is performed by induction in $k$ and $N$. For $k=1$ the identity (\ref{e51}) can be deduced from the higher order $R$-matrix identities based on the associative Yang-Baxter equation.

\subsection{Proof of the Theorem for $k=1$}
The proof of (\ref{e02}) can be performed straightforwardly using (\ref{x1}), which is a higher order generalization of the addition formula (\ref{Fay}).
 Details can be found in the
Appendix of \cite{ZZ2}. Here we prove (\ref{e53}) for $k=1$ in a similar way. For this purpose we recall that
the elliptic $R$-matrix in the fundamental representation (\ref{BB}) satisfies the so-called associative Yang-Baxter equation \cite{FK}:
\beq\label{AYBE}
\begin{array}{c}
    R^{u}_{12} R^{u'}_{23} = R^{u'}_{13} R^{u-u'}_{12} + R^{u'-u}_{23} R^{u}_{13}, \quad R^\hbar_{ab} = R^\hbar_{ab}(z_a-z_b)\,.
\end{array}
\eq
In \cite{Z18} the higher order analogue of (\ref{AYBE}) was derived. Here we use it in a slightly different form suggested in \cite{MZ1}:
 \beq\label{a4}
  \begin{array}{c}
  \displaystyle{
 \overrightarrow{\prod\limits_{i=1}^n}
 R_{a,i}^{u_i}(w_i)=
 \sum\limits_{m=1}^n
 \overrightarrow{\prod\limits_{j=m+1}^{n}}
 R_{m,j}^{u_j}(w_j-w_{m})
 \cdot R_{a,m}^U(w_m)\cdot
  \overrightarrow{\prod\limits_{j=1}^{m-1}}
 R_{m,j}^{u_j}(w_j-w_{m})\,,
 }
 \end{array}
 \eq
where $U=\sum\limits_{k=1}^n u_k$ and $a\notin\{1...n\}$, $n\in\mZ_+$. Let us fix $a=1$
together with replacing the set $\{1...n\}$ with $\{2...n+1\}$. Then (\ref{a4}) takes the form
 \beq\label{e54}
  \begin{array}{c}
  \displaystyle{
 \overrightarrow{\prod\limits_{i=2}^{n+1}}
 R_{1,i}^{u_i}(w_i)=
 \sum\limits_{m=2}^{n+1}
 \overrightarrow{\prod\limits_{j=m+1}^{n+1}}
 R_{m,j}^{u_j}(w_j-w_{m})
 \cdot R_{1,m}^U(w_m)\cdot
  \overrightarrow{\prod\limits_{j=2}^{m-1}}
 R_{m,j}^{u_j}(w_j-w_{m})\,.
 }
 \end{array}
 \eq
Set $n=2N-1$ and identify $w_i=z_1-z_i$ for $i=2,...,2N$ (recall (\ref{e51})-(\ref{e52})).
Also, put $u_2=...=u_N=\hbar$ and $u_{N+1}=...=u_{2N}=-\hbar$. Then $U=-\hbar$.
After all substitutions in the l.h.s. of (\ref{e54}) one gets $R_{12}^\hbar...R^\hbar_{1N}\cdot R^{-\hbar}_{1,N+1}...R^{-\hbar}_{1,2N}$. Writing down the r.h.s. of (\ref{e54}) finally yields
 \beq\label{e55}
  \begin{array}{c}
  \displaystyle{
  \sum\limits_{i=1}^N R^\hbar_{i,i+1}...R^\hbar_{i,N}\cdot
  R^\hbar_{N+1,i}...R^\hbar_{2N,i}\cdot R^\hbar_{i,1}...R^\hbar_{i,i-1}
  -
  }
  \\ \ \\
 \displaystyle{
  -
  \sum\limits_{i=1}^N
  R^\hbar_{N+i+1,N+i}...R^\hbar_{2N,N+i} \cdot
  R^\hbar_{N+i,1}...R^\hbar_{N+i,N}\cdot
  R^\hbar_{N+1,N+i}...R^\hbar_{N+i-1,N+i}=0\,,
 }
 \end{array}
 \eq
where we used the skew-symmetry of $R$-matrices (\ref{r08}). The obtained relation (\ref{e55}) is (\ref{e53}) for $k=1$.

\subsection{Auxiliary lemmas}
Here we formulate two lemmas useful for the proof of the Theorem. The proof of Lemmas is given in the Appendix.

\begin{lemma}\label{Lem31}
Consider $\mF=\mF[k,N]$ as function of $z_a=x_a$. It has no poles at $x_1,\ldots x_{a-1},x_{a+1}\ldots,x_N$.
\end{lemma}

\begin{lemma}\label{Lem32}
Consider $\mF=\mF[k,N]$ as function of $z_a=x_a$. The following relation holds for any $b\in{\mathcal N}_2$:
 \beq\label{e58}
  \begin{array}{c}
  \displaystyle{
  \res\limits_{z_a=z_{b}}\mF[k,N]=-G\cdot \mF[k-1,N-1]\cdot H\,,\quad b=N+1,\ldots,2N\,,
 }
 \end{array}
 \eq
where
 \beq\label{e59}
  \begin{array}{c}
  \displaystyle{
  G=R^\hbar_{a,a+1}\ldots R^\hbar_{a,N} \cdot R^\hbar_{b+1,b}\ldots R^\hbar_{2N,b}\,,
 }
 \end{array}
 \eq
 \beq\label{e60}
  \begin{array}{c}
  \displaystyle{
  H=R^\hbar_{N+1,a}\ldots R_{b-1,a}^\hbar\cdot R^\hbar_{b,1}\ldots R^\hbar_{b,a-1}P_{ab}
 }
 \end{array}
 \eq
and by $\mF[k-1,N-1]$ we mean the expression $\mF=\mF_1-\mF_2$ (\ref{e531})-(\ref{e532}) taken for $k-1$ and
depending on $2(N-1)$ variables $\{x_1,...,x_N\}\setminus\{x_a\}$, $\{y_1,...,y_N\}\setminus\{y_b\}$.
\end{lemma}
The latter statement is applied to the proof of the Theorem by setting $a=1$.

\subsection{Proof of the Theorem for $k>1$}
Let us prove the main Theorem (\ref{e53}).

\begin{proof}
The proof is by induction in $k$. For $k=1$ the statement was proved above in (\ref{e55}).
Suppose (\ref{e53}) is true for $k-1$. We need to prove it for $k$.

Main idea is the same as in the scalar case. We are going to use Lemma \ref{Lem41} given in the Appendix. For this purpose
consider the quasi-periodic properties of $\mF$ as function of $z_1$. Due to (\ref{r721}) any $R$-matrix of the form $R^\hbar_{1k}=R^\hbar_{1k}(z_1,z_k)=R^\hbar_{1k}(z_1-z_k)$, $k=2,\ldots,2N$ is transformed as follows:
 \beq\label{e61}
  \begin{array}{c}
  \displaystyle{
  R^\hbar_{1k}(z_1+1,z_k)=Q_1^{-1}R_{1k}^\hbar(z_1,z_k)Q_1\,,
 }
 \\ \ \\
  \displaystyle{
 R^\hbar_{1k}(z_1+\tau,z_k)=\exp(-2\pi\imath\frac{\hbar}{M})\,\Lambda_1^{-1}R_{1k}^\hbar(z_1,z_k)\Lambda_1\,,
  }
 \end{array}
 \eq
where $Q_1=Q\otimes 1_M\otimes\ldots\otimes 1_M=Q\otimes 1_M^{\otimes(2N-1)}$ and similarly for $\Lambda_1$. For $R$-matrices of the form $R^\hbar_{k1}=R^\hbar_{k1}(z_k,z_1)=R^\hbar_{k1}(z_k-z_1)$ we have:
 \beq\label{e62}
  \begin{array}{c}
  \displaystyle{
  R^\hbar_{k1}(z_k,z_1+1)=R^\hbar_{k1}(z_k-1,z_1)=Q_k R_{k1}^\hbar(z_1)Q_k^{-1}=Q_1^{-1}R_{12}^\hbar(z_k,z_1)Q_1
 }
 \end{array}
 \eq
and in the same way
 \beq\label{e63}
  \begin{array}{c}
  \displaystyle{
 R^\hbar_{k1}(z_k,z_1+\tau)=\exp(2\pi\imath\frac{\hbar}{M})\,\Lambda_1^{-1}R_{k1}^\hbar(z_k,z_1)\Lambda_1\,,
  }
 \end{array}
 \eq
Thus, any $R$-matrix which contains index $1$ is transformed as (\ref{e61}) or (\ref{e62})-(\ref{e63}). The number of
$R$-matrices of the forms $R^\hbar_{1k}$ or $R^\hbar_{k1}$ is different in different terms. But the difference of these numbers is the same for all the terms: the number of
$R$-matrices of the form $R^\hbar_{1k}$ is less than those of the form $R^\hbar_{k1}$ by $k$. Therefore,
 \beq\label{e64}
  \begin{array}{c}
  \displaystyle{
  \mF(z_1+1)=Q_1^{-1}\mF(z_1)Q_1\,,
 }
 \\ \ \\
  \displaystyle{
 \mF(z_1+\tau)=\exp(2\pi\imath\frac{k\hbar}{M})\,\Lambda_1^{-1}\mF(z_1)\Lambda_1\,.
  }
 \end{array}
 \eq
In order to use Lemma \ref{Lem41} let us extend the torus to the large one with periods $M$, $M\tau$. Then
due to $Q^M=\Lambda^M=1_M$ we have
 \beq\label{e65}
  \begin{array}{c}
  \displaystyle{
  \mF(z_1+M)=\mF(z_1)\,,
 }
 \\ \ \\
  \displaystyle{
 \mF(z_1+M\tau)=\exp(2\pi\imath\, k\hbar)\mF(z_1)\,.
  }
 \end{array}
 \eq
Next, notice that $\mF(z_1)$ is an entire function of $z_1$ on the fundamental parallelogram ($0\leq z_1<1$, $0\leq z_1<\tau$). Indeed, due to (\ref{r05}) all possible poles
inside the fundamental parallelogram are
at points $x_2,\ldots,x_N$ and $y_1,\ldots,y_N$. By Lemma \ref{Lem31} the poles at points $x_2,\ldots,x_N$ are absent. The absence of poles $y_1,\ldots,y_N$ easily follows from Lemma \ref{Lem32} and the induction assumption
(which implies $\mF[k-1,N-1]=0$). Finally, by taking residue $\res\limits_{z_1=z_k}$, $k=2,\ldots,2N$ of both sides of (\ref{e64}) we conclude that the absence of poles in the fundamental parallelogram is extended to the absence of poles on the large torus. Then all requirements of Lemma \ref{Lem41} are fulfilled and the Theorem is proved.
\end{proof}

\subsection{Examples}
Let us give a few explicit examples for the identities (\ref{e021}).
\paragraph{Example. $N=2,\ k=1$:}
  \beq\label{e70}
R_{12}^\hbar R_{31}^{\hbar} R_{41}^{\hbar}+R_{32}^{\hbar} R_{42}^{\hbar}R_{21}^\hbar-R_{43}^{\hbar} R_{31}^{\hbar}R_{32}^{\hbar}- R_{41}^{\hbar}R_{42}^{\hbar} R_{34}^{\hbar}=0\,.
 \eq
 \paragraph{Example. $N=3,\ k=1$:}
 \beq\label{e71}
  \begin{array}{lll}
 R_{12}^\hbar  R_{13}^\hbar  R_{41}^{\hbar} R_{51}^{\hbar} R_{61}^{\hbar}+ R_{23}^\hbar   R_{42}^{\hbar} R_{52}^{\hbar} R_{62}^{\hbar} R_{21}^\hbar+ R_{43}^{\hbar} R_{53}^{\hbar} R_{63}^{\hbar} R_{31}^\hbar R_{32}^\hbar\\
 \ \\
 -R_{54}^{\hbar} R_{64}^{\hbar} R_{41}^{\hbar} R_{42}^{\hbar} R_{43}^{\hbar}   -R_{65}^{\hbar} R_{51}^{\hbar} R_{52}^{\hbar} R_{53}^{\hbar}   R_{45}^{\hbar} -R_{61}^{\hbar} R_{62}^\hbar R_{63}^{\hbar} R_{46}^{\hbar} R_{56}^{\hbar} =0 \,.
 \end{array}
 \eq
  \paragraph{Example. $N=3,\ k=2$:}
 \beq\label{e72}
  \begin{array}{lll}
 R_{23}^\hbar  R_{13}^\hbar R_{41}^{\hbar} R_{51}^{\hbar} R_{61}^{\hbar} R_{42}^{\hbar} R_{52}^{\hbar} R_{62}^{\hbar}+ R_{12}^\hbar   R_{41}^{\hbar} R_{51}^{\hbar} R_{61}^{\hbar}  R_{43}^{\hbar} R_{53}^{\hbar} R_{63}^{\hbar} R_{32}^\hbar
 \\ \ \\
 + R_{42}^{\hbar} R_{52}^{\hbar} R_{62}^{\hbar}  R_{43}^{\hbar} R_{53}^{\hbar} R_{63}^{\hbar}R_{31}^\hbar R_{21}^\hbar
 -R_{65}^{\hbar} R_{64}^{\hbar} R_{41}^{\hbar}R_{42}^{\hbar} R_{43}^{\hbar} R_{51}^{\hbar} R_{52}^\hbar R_{53}^\hbar
 \\ \ \\
 -R_{54}^{\hbar} R_{41}^{\hbar} R_{42}^{\hbar} R_{43}^{\hbar}R_{61}^{\hbar} R_{62}^{\hbar} R_{63}^{\hbar}  R_{56}^{\hbar} -R_{51}^\hbar R_{52}^\hbar R_{53}^{\hbar} R_{61}^{\hbar} R_{62}^{\hbar} R_{63}^{\hbar} R_{46}^{\hbar} R_{45}^{\hbar} =0 \,.
 \end{array}
 \eq

Let us make a comment on the trigonometric and rational degenerations. All trigonometric and rational $R$-matrices
considered in \cite{MZ1,MZ2} satisfy the obtained set of identities. Details will be given in a separate publication.

Another remark is that in trigonometric and rational cases the scalar identities (\ref{e01}) can be extended to the case when the number of variables ${\bf x}$ is different from the number of variables ${\bf y}$, so that we have $|{\bf x}|=N_1$, $|{\bf y}|=N_2$ and $N_1\neq N_2$. This can be explained as follows. In the trigonometric or rational case one can consider a limit $x_i\rightarrow\infty$, which decrease the number of variables ${\bf x}$ by one. Similar arguments work for $R$-matrix identities if the limit $\lim\limits_{x_i\rightarrow \infty}R_{ij}(x_i-x_j)$
(or $\lim\limits_{x_i\rightarrow \infty}R_{i,N+j}(x_i-y_j)$) is well defined. It is not true for all possible degenerations of the elliptic $R$-matrix. At the same time it is true for a number of widely known cases. For example, it is obviously true for the rational Yang's $R$-matrix $R_{12}^\hbar(z)=1_M\otimes1_M\hbar^{-1}+P_{12}z^{-1}$. Then $\lim\limits_{z\rightarrow \infty}R^\hbar_{12}(z)=1_M\otimes1_M\hbar^{-1}$. Therefore, the limit $x_i\rightarrow\infty$ being applied to some identity provides another identity, where all $R$-matrices containing $x_i$ are replaced by the scalar factor $\hbar^{-1}$. We will consider these cases in our future works.

\section{Appendix A: elliptic functions}\label{sectA}
\def\theequation{A.\arabic{equation}}
\setcounter{equation}{0}

\paragraph{Elliptic functions.} We use the following theta-function:
\beq\label{a0963}\begin{array}{c} \displaystyle{
     \vartheta (z)=\vartheta (z|\tau) = -\sum_{k\in \mathbb{Z}} \exp \left( \pi \imath \tau (k + \frac{1}{2})^2 + 2\pi \imath (z + \frac{1}{2}) (k + \frac{1}{2}) \right)\,,\quad {\rm Im}(\tau)>0\,.
}\end{array}\eq
It is odd $\vth(-z)=-\vth(z)$ and has simple zero at $z=0$.
The  Kronecker elliptic function is defined as
\beq\displaystyle{
\label{a0962}
    \phi(z, u) =
            \frac{\vartheta'(0) \vartheta (z + u)}{\vartheta (z) \vartheta (u)}=\phi(u,z)\,.
}\eq
 As function of $z$ it has simple pole at $z=0$ and
$
\res\limits_{z=0}\phi(z,u)=1\,.
$
The quasi-periodic properties on the lattice of periods $\Gamma=\mZ\oplus\mZ\tau$ (of elliptic curve $\mC/\Gamma$):
 \beq\label{a0961}
  \begin{array}{l}
  \displaystyle{
 \vth(z+1)=-\vth(z)\,,\qquad \vth(z+\tau)=-e^{-\pi\imath\tau-2\pi\imath z}\vth(z)\,,
 }
 \end{array}
 \eq
 so that
 \beq\label{a096}
  \begin{array}{l}
  \displaystyle{
 \phi(z+1,u)= \phi(z,u)\,,\qquad  \phi(z+\tau,u)= e^{-2\pi\imath u}\phi(z,u)\,.
 }
 \end{array}
 \eq
 The function (\ref{a0962}) satisfies the following addition formula:
%
\beq \begin{array}{c} \label{Fay} \displaystyle{
    \phi(z_1, u_1) \phi(z_2, u_2) = \phi(z_1, u_1 + u_2) \phi(z_2 - z_1, u_2) + \phi(z_2, u_1 + u_2) \phi(z_1 - z_2, u_1)
}\end{array}\eq
and the identity
\beq\begin{array}{c} \displaystyle{
\label{a0964}
    \phi(z, u) \phi(z, -u) = \wp (z) - \wp (u)\,.
}\end{array}\eq
where $\wp (x)$ -- is the Weierstrass $\wp$-function. The higher order analogue of (\ref{Fay}) is as follows:
\beq\label{x1}
  \begin{array}{c}
  \displaystyle{
 \prod\limits_{i=1}^n \phi(w_i,u_i)=\sum\limits_{i=1}^n
 \phi \Bigl (w_i,\sum\limits_{l=1}^nu_l \Bigr )\prod\limits_{j\neq
 i}^n\phi(w_j-w_i,u_j)\,,\quad n\in\mZ_+\,.
 }
 \end{array}
 \eq

\paragraph{R-matrix.} Define the set of $M^2$ functions:
 \beq\label{a08}
 \begin{array}{c}
  \displaystyle{
 \vf_a(z,\om_a+\hbar)=\exp(2\pi\imath\frac{a_2z}{M})\,\phi(z,\om_a+\hbar)\,,\quad
 \om_a=\frac{a_1+a_2\tau}{M}\,,
 }
 \end{array}
 \eq
where $M\in\mZ_+$ is an integer number and $a=(a_1, a_2)\in\mZ_M\times\mZ_M$.

For construction of the elliptic Baxter-Belavin $R$-matrix \cite{Baxter,Belavin} (see also \cite{LOZ15}) introduce the following matrix basis in $\MatM$:
\beq\label{a971}\begin{array}{c} \displaystyle{
        T_\al = \exp \left( \al_1 \al_2 \frac{\pi \imath}{M} \right) Q^{\al_1} \Lambda^{\al_2}, \quad \al = (\al_1, \al_2)\in \mZ_M \times \mZ_M\,,
}\end{array}\eq
where $Q,\Lambda\in\MatM$ are given by
\beq\label{a041}
 \begin{array}{c}
  \displaystyle{
 Q_{kl}=\delta_{kl}\exp(\frac{2\pi
 \imath}{{ M}}k)\,,\ \ \
 \Lambda_{kl}=\delta_{k-l+1=0\,{\hbox{\tiny{mod}}}\,
 { M}}\,,\quad Q^{ M}=\Lambda^{ M}=1_M\,.
 }
 \end{array}
 \eq
 For example, $T_0=T_{(0,0)}=1_M$ is the identity $M\times M$ matrix.
 The matrices $Q$ and $\Lambda$ satisfy relations
 \beq\label{a051}
 \begin{array}{c}
  \displaystyle{
 \Lambda^{a_2} Q^{a_1}=\exp\left(\frac{2\pi\imath}{{ M}}\,a_1
 a_2\right)Q^{a_1} \Lambda^{a_2}\,,\ a_{1,2}\in\mZ\,.
 }
 \end{array}
 \eq
 Then
$
    T_\al T_\be = \ka_{\al, \be} T_{\al + \be}$,
    $\ka_{\al, \be} = \exp \left( \frac{\pi i}{M}(\al_2 \be_1 - \al_1 \be_2) \right)
$,
where $\al+\be=(\al_1+\be_1,\al_2+\be_2)$.

Using (\ref{a08}) and (\ref{a971}) introduce the Baxter-Belavin elliptic ($\mZ_M$ symmetric) $R$-matrix:
\begin{equation}\label{BB}
\begin{array}{c}
    \displaystyle{
    R^{\hbar}_{12} (x) = \frac{1}{M}
    \sum_\al \varphi_\al (x, \frac{\hbar}{M} + \om_\al) T_\al \otimes T_{-\al}}\in\MatM^{\otimes 2}\,.
\end{array}
\end{equation}
%
The $\mZ_M$ symmetry means
that
\begin{equation}\label{r081}
\begin{array}{c}
    \displaystyle{
    (Q\otimes Q) R^{\hbar}_{12} (x) = R^{\hbar}_{12} (x) (Q\otimes Q) \,,\qquad
    (\Lambda\otimes \Lambda) R^{\hbar}_{12} (x) = R^{\hbar}_{12} (x) (\Lambda\otimes \Lambda) \,.
    }
\end{array}
\end{equation}
%
%
%
  The $R$-matrix (\ref{BB}) has simple poles in both variables ($\hbar$ or $z$) with the following residues:
 \beq\label{r05}
 \begin{array}{c}
  \displaystyle{
\res\limits_{z=0}R^{\hbar}_{12}(z)=
P_{12}\,,
\qquad
\res\limits_{\hbar=0} R^{\hbar}_{12}(z)=1_M\otimes 1_M\,.
}
 \end{array}
 \eq
 The quasi-periodic behaviour on the lattice of periods is as follows:
 \beq\label{r721}
 \begin{array}{c}
  \displaystyle{
 R_{12}^\hbar(z+1)=(Q^{-1}\otimes 1_M)R_{12}^\hbar(z)(Q\otimes 1_M)\,,
  }
 \\ \ \\
  \displaystyle{
 R_{12}^\hbar(z+\tau)=\exp(-2\pi\imath\frac{\hbar}{M})\,(\Lambda^{-1}\otimes 1_M)R_{12}^\hbar(z)(\Lambda\otimes
 1_M)\,.
  }
 \end{array}
 \eq
 More properties can be found in the Appendix of \cite{MZ1}.

 \paragraph{Proof of identities.} Main idea in proving identities from this paper is formulated in the Lemma below. The same Lemma was used in \cite{Ruij2006}. Although it is widely known in the theory of elliptic functions we briefly recall its proof for readers convenience.

\begin{lemma}\label{Lem41}
 Let $f(z)$ be an entire function on elliptic curve $\Sigma_\tau=\mC/(\mZ\oplus\tau\mZ)$ with the following quasi-periodic properties on the lattice of periods $\mZ\oplus\tau\mZ$:
 \beq\label{e901}
 \begin{array}{c}
  \displaystyle{
  f(z+1)=f(z)\,,
  }
 \\ \ \\
  \displaystyle{
 f(z+\tau)=e^{2\pi \imath \al}f(z)\,,\quad \al\in\mC\,.
  }
 \end{array}
 \eq
 Then this function is equal to zero if $\al\notin \mZ\oplus\tau\mZ$ (that is $\al\neq n_1+n_2\tau$, $n_1,n_2\in\mZ$).
\end{lemma}
\begin{proof}
Consider the function
 \beq\label{e902}
 \begin{array}{c}
  \displaystyle{
  {\ti f}(z)=e^{2\pi\imath \be z}f(z)\,,\quad \be\in{\mathbb R}\,.
  }
 \end{array}
 \eq
 It has the following properties:
 \beq\label{e903}
 \begin{array}{c}
  \displaystyle{
  {\ti f}(z+1)=e^{2\pi\imath\be}{\ti f}(z)\,,
  }
 \\ \ \\
  \displaystyle{
 {\ti f}(z+\tau)=e^{2\pi \imath (\al+\be\tau)}{\ti f}(z)\,.
  }
 \end{array}
 \eq
 By fixing the real valued parameter $\be$ as
 \beq\label{e904}
 \begin{array}{c}
  \displaystyle{
  \be=-\frac{{\rm Im}(\al)}{{\rm Im}(\tau)}=-\frac{\al-\bar\al}{\tau-\bar\tau}
  }
 \end{array}
 \eq
 we come to
 \beq\label{e905}
 \begin{array}{c}
  \displaystyle{
  {\ti f}(z+1)=g_1{\ti f}(z)\,,\quad  g_1=\exp\Big(-2\pi\imath\frac{\al-\bar\al}{\tau-\bar\tau}\Big)\,,
  }
 \\ \ \\
  \displaystyle{
 {\ti f}(z+\tau)=g_\tau{\ti f}(z)\,,\quad g_\tau=\exp\Big(2\pi \imath ({\rm Re}(\al)+\be{\rm Re}(\tau))\Big)=
 \exp\Big(2\pi\imath\frac{\tau\bar\al-\al\bar\tau}{\tau-\bar\tau}\Big)\,.
  }
 \end{array}
 \eq
 In the case $\al=n_1+n_2\tau$ for some $n_1,n_2\in\mZ$ the function ${\ti f}$ becomes double periodic, i.e. $g_1=g_\tau=1$.
 Any entire double-periodic function on elliptic curve is a constant (since elliptic curve is compact). Therefore,
 due to (\ref{e902}), (\ref{e904}) we obtain $f(z)={\rm const}\cdot\exp(-2\pi\imath z\frac{\al-\bar\al}{\tau-\bar\tau})$\,.

 If $\al\notin \mZ\oplus\tau\mZ$ then $g_1\neq 1$ and/or $g_\tau\neq 1$. At the same time $|g_1|=|g_\tau|=1$. Then due to the Liouville theorem (the one which says that every bounded entire holomorphic function on elliptic curve is constant) we conclude: $\ti f(z)={\rm const}$. The latter constant equals zero since $g_1\neq 1$ and/or $g_\tau\neq 1$.
\end{proof}


\section{Appendix B: proof of auxiliary lemmas}\label{sectB}
\def\theequation{B.\arabic{equation}}
\setcounter{equation}{0}

\subsection*{Proof of Lemma \ref{lemmaY}}

\begin{proof}
Consider the l.h.s. of  (\ref{lemmaY1}):
\beq\label{Y10}
 {\mathcal Y}^{\hbar}_{I,J}R^{\hbar}_{b,j_{n+1}}\dots R^{\hbar}_{b,j_l}=\overrightarrow{\prod_{j\in J}}R_{i_1,j}^{\hbar}\ldots \overrightarrow{\prod_{j\in J}}R_{i_{k},j}^{\hbar}\cdot R^{\hbar}_{b,j_{n+1}} \dots R^{\hbar}_{b,j_l}\,.
\eq
Move $R^{\hbar}_{b,j_{n+1}}$ to the left using the Yang-Baxter equation:
\beq\label{Y11}
\begin{array}{lll}
\displaystyle{
\overrightarrow{\prod_{j\in J}}R_{i_{k},j}^{\hbar}\cdot R^{\hbar}_{b,j_{n+1}}=\left(R^{\hbar}_{i_k,j_1}\dots R^{\hbar}_{i_k,j_{n-1}}R^{\hbar}_{i_k,b}R^{\hbar}_{i_k,j_{n+1}} \dots R^{\hbar}_{i_k,j_{l}}\right)R^{\hbar}_{b,j_{n+1}}=
}
\\ \ \\
=R^{\hbar}_{i_k,j_1}\dots R^{\hbar}_{i_k,j_{n-1}}R^{\hbar}_{i_k,b}R^{\hbar}_{i_k,j_{n+1}}R^{\hbar}_{b,j_{n+1}} \dots R^{\hbar}_{i_k,j_{l}}
\\ \ \\
=R^{\hbar}_{i_k,j_1}\dots R^{\hbar}_{i_k,j_{n-1}}R^{\hbar}_{b,j_{n+1}}R^{\hbar}_{i_k,j_{n+1}}R^{\hbar}_{i_k,b} \dots R^{\hbar}_{i_k,j_{l}}=
\\ \ \\
=R^{\hbar}_{b,j_{n+1}}R^{\hbar}_{i_k,j_1}\dots R^{\hbar}_{i_k,j_{n-1}}R^{\hbar}_{i_k,j_{n+1}}R^{\hbar}_{i_k,b} \dots R^{\hbar}_{i_k,j_{l}}\,.
\end{array}
\eq
In the second and the last line of (\ref{Y11}) we used (\ref{QYB3}), and in the third line we used (\ref{QYB2}). In the same way, by moving the product $ R^{\hbar}_{b,j_{n+1}} \dots R^{\hbar}_{b,j_l}$ to the left through $\displaystyle\overrightarrow{\prod_{j\in J}}R_{i_{k},j}^{\hbar}$ one obtains:
\beq\label{Y12}
\overrightarrow{\prod_{j\in J}}R_{i_{k},j}^{\hbar}\cdot R^{\hbar}_{b,j_{n+1}} \dots R^{\hbar}_{b,j_l}= R^{\hbar}_{b,j_{n+1}} \dots R^{\hbar}_{b,j_l}\cdot \overrightarrow{\prod\limits_{\substack {j\in J \\ j\neq b}}}R_{i_{k},j}^{\hbar}\cdot R^{\hbar}_{i_k,b} \,.
\eq
By applying (\ref{Y12}) to each product $\displaystyle\overrightarrow{\prod_{j\in J}}R_{i_l,j}^{\hbar}$ in (\ref{Y10})  and by moving each $R^{\hbar}_{i_l,b}$ to the right we get (\ref{lemmaY1}).
  \end{proof}

\subsection*{Proof of Lemma \ref{Lem31}}

\begin{proof}
  Arguments of $R$-matrices of the form $x_a-x_b$ appear in $\mF_1$, and there are no such arguments in $\mF_2$. Therefore,
 \beq\label{e56}
  \begin{array}{c}
  \displaystyle{
  \res\limits_{x_a=x_b}\mF=\res\limits_{x_a=x_b}\mF_1\,,\quad b=1,\ldots a-1,a+1\ldots,N\,,
 }
 \end{array}
 \eq
   and we need to prove\footnote{We should notice that the structure of terms in the sum $\mF_1$ excludes appearance
   of higher order poles. This is why we are looking for simple poles only.}
 \beq\label{e57}
  \begin{array}{c}
  \displaystyle{
  \res\limits_{x_a=x_b}\mF_1=0\,.
 }
 \end{array}
 \eq
  Consider the expression $\mF_1$ (\ref{e531}). Due to the structure of poles (\ref{r05}) a nonzero input
  into the l.h.s. of (\ref{e57}) is provided by the terms (in $\mF_1$) which contain $R^\hbar_{ab}$ or $R^\hbar_{ba}$. It happens in two cases:

  1. $a\in I$ and $b\in I^c$;

  2. $b\in I$ and $a\in I^c$.

 \noindent Let $a<b$. Fix some set $J$ of $(k-1)$ elements which does not contain $a$ and $b$. Denote by $J^\bullet=J^c\setminus\{a,b\}$. In the first case we choose $I=J\cup\{a\}$ and $I^c=J^\bullet\cup\{b\}$ and in the second case $I=J\cup\{b\}$ and $I^c=J^\bullet\cup\{a\}$.
We are going to show that the residue of the first one is exactly cancelled out by
 the residue of the second one:
 \beq\label{e571}
  \begin{array}{c}
 \res\limits_{z_a=z_b} \mathcal{R}_{J\cup\{a\},J^\bullet\cup\{b\}}^{\hbar}\cdot{\mathcal Y}^{\hbar}_{\mathcal N_2,J\cup\{a\}} \cdot\mathcal{R'}_{J\cup\{a\},J^\bullet\cup\{b\}}^{\hbar}=
 \\ \ \\
 -\res\limits_{z_a=z_b} \mathcal{R}_{J\cup\{b\},J^\bullet\cup\{a\}}^{\hbar}\cdot{\mathcal Y}^{\hbar}_{\mathcal N_2,J\cup\{b\}} \cdot\mathcal{R'}_{J\cup\{b\},J^\bullet\cup\{a\}}^{\hbar}\,.
 \end{array}
 \eq
In order to prove this we separate in the l.h.s. of  (\ref{e571}) those $R$-matrices which contain index $a$ or $b$ using relations (\ref{l21})-(\ref{l22}) and (\ref{lemmaY1}):

\beq\label{e572}
\mathcal{R}_{J\cup\{a\},J^\bullet\cup\{b\}}^{\hbar}=\mathcal{R}_{\{a\},J\cup J^\bullet\cup\{b\}}^{\hbar}\cdot\mathcal{R}_{J\cup J^\bullet,\{b\}}^{\hbar}\cdot\mathcal{R}_{J,J^\bullet}^{\hbar}\left(\mathcal{R}_{J^\bullet,\{b\}}^{\hbar}\right)^{-1} \left(\mathcal{R}_{\{a\},J}^{\hbar}\right)^{-1}\,,
\eq

\beq\label{e573}
{\mathcal Y}^{\hbar}_{\mathcal N_2,J\cup\{a\}}=R_{a,j_{n+1}}^{\hbar}\ldots R_{a,j_{k-1}}^{\hbar}\cdot{\mathcal Y}^{\hbar}_{\mathcal N_2,J}\cdot{\mathcal Y}^{\hbar}_{\mathcal N_2,\{a\}}\cdot\left(R_{a,j_{n+1}}^{\hbar}\ldots R_{a,j_{k-1}}^{\hbar}\right)^{-1}\,,
\eq

\beq\label{e574}
\mathcal{R'}_{J\cup\{a\},J^\bullet\cup\{b\}}^{\hbar}=\left(\mathcal{R'}_{\{a\},J}^{\hbar}\right)^{-1}\left(\mathcal{R'}_{J^\bullet,\{b\}}^{\hbar}\right)^{-1}\mathcal{R'}_{J,J^\bullet}^{\hbar}\cdot\mathcal{R'}_{J\cup J^\bullet,\{b\}}^{\hbar}\cdot \mathcal{R'}_{\{a\},J\cup J^\bullet\cup\{b\}}^{\hbar}\,.
\eq
Then using (\ref{e572})-(\ref{e574}) we simplify the l.h.s. of (\ref{e571}) before taking the residue:
\beq\label{e575}
\begin{array}{lll}
\mathcal{R}_{J\cup\{a\},J^\bullet\cup\{b\}}^{\hbar}{\mathcal Y}^{\hbar}_{\mathcal N_2,J\cup\{a\}}\mathcal{R'}_{J\cup\{a\},J^\bullet\cup\{b\}}^{\hbar}=\mathcal{R}_{\{a\},J\cup J^\bullet\cup\{b\}}^{\hbar}\cdot\mathcal{R}_{J\cup J^\bullet,\{b\}}^{\hbar}\cdot\mathcal{R}_{J,J^\bullet}^{\hbar}\left(\mathcal{R}_{J^\bullet,\{b\}}^{\hbar}\right)^{-1}
\\ \ \\
{\mathcal Y}^{\hbar}_{\mathcal N_2,J}\cdot{\mathcal Y}^{\hbar}_{\mathcal N_2,\{a\}}\left(\mathcal{Y}^{\hbar}_{\{a\},J}\right)^{-1}\cdot\left(\mathcal{R'}_{J^\bullet,\{b\}}^{\hbar}\right)^{-1}\mathcal{R'}_{J,J^\bullet}^{\hbar}\cdot\mathcal{R'}_{J\cup J^\bullet,\{b\}}^{\hbar}\cdot \mathcal{R'}_{\{a\},J\cup J^\bullet\cup\{b\}}^{\hbar}\,.
\end{array}
\eq
Here we reduced $\left(\mathcal{R}_{\{a\},J}^{\hbar}\right)^{-1}$ from (\ref{e572}) and $R_{a,j_{n+1}}^{\hbar}\ldots R_{a,j_{k-1}}^{\hbar}$ from (\ref{e573}) and also used that
$$\displaystyle \left(R_{a,j_{n+1}}^{\hbar}\ldots R_{a,j_{k-1}}^{\hbar}\right)^{-1}\left(\mathcal{R'}_{\{a\},J}^{\hbar}\right)^{-1}=\left(\mathcal{Y}^{\hbar}_{\{a\},J}\right)^{-1}.$$
Next, we move $\left(\mathcal{R}_{J^\bullet,\{b\}}^{\hbar}\right)^{-1} $ to the right through ${\mathcal Y}^{\hbar}_{\mathcal N_2,J}{\mathcal Y}^{\hbar}_{\mathcal N_2,\{a\}}\left(\mathcal{Y}^{\hbar}_{\{a\},J}\right)^{-1}$ since it contains pairwise distinct indices:
\beq\label{e576}
\begin{array}{lll}
\mathcal{R}_{J\cup\{a\},J^\bullet\cup\{b\}}^{\hbar}{\mathcal Y}^{\hbar}_{\mathcal N_2,J\cup\{a\}}\mathcal{R'}_{J\cup\{a\},J^\bullet\cup\{b\}}^{\hbar}=\mathcal{R}_{\{a\},J\cup J^\bullet\cup\{b\}}^{\hbar}\cdot\mathcal{R}_{J\cup J^\bullet,\{b\}}^{\hbar}\cdot\mathcal{R}_{J,J^\bullet}^{\hbar}
\\ \ \\
{\mathcal Y}^{\hbar}_{\mathcal N_2,J}\cdot{\mathcal Y}^{\hbar}_{\mathcal N_2,\{a\}}\left(\mathcal{Y}^{\hbar}_{\{a\},J}\right)^{-1}\cdot\left(\mathcal{R}_{J^\bullet,\{b\}}^{\hbar}\right)^{-1} \left(\mathcal{R'}_{J^\bullet,\{b\}}^{\hbar}\right)^{-1}\mathcal{R'}_{J,J^\bullet}^{\hbar}\cdot\mathcal{R'}_{J\cup J^\bullet,\{b\}}^{\hbar} \cdot\mathcal{R'}_{\{a\},J\cup J^\bullet\cup\{b\}}^{\hbar}\,.
\end{array}
\eq
Now we take the residue at $z_a=z_b$ by replacing  $R_{a,b}^{\hbar}$ in  $\mathcal{R}_{\{a\},J\cup J^\bullet\cup\{b\}}^{\hbar}$ with the permutation operator $P_{ab}$ and move it to the right:
\beq\label{e577}
\begin{array}{c}
\res\limits_{z_a=z_b}\mathcal{R}_{J\cup\{a\},J^\bullet\cup\{b\}}^{\hbar}{\mathcal Y}^{\hbar}_{\mathcal N_2,J\cup\{a\}}\mathcal{R'}_{J\cup\{a\},J^\bullet\cup\{b\}}^{\hbar}=
\\ \ \\
R_{a,a+1}^{\hbar}\ldots R_{a,b-1}^{\hbar}P_{ab}R_{a,b+1}^{\hbar}\dots R_{a,N}^{\hbar}\cdot R_{b-1,b}^{\hbar}\ldots R_{a+1,b}^{\hbar} R_{a-1,b}^{\hbar}\ldots R_{1,b}^{\hbar}\cdot\mathcal{R}_{J,J^\bullet}^{\hbar}\cdot{\mathcal Y}^{\hbar}_{\mathcal N_2,J}\cdot{\mathcal Y}^{\hbar}_{\mathcal N_2,\{a\}}
\\ \ \\
\left(\mathcal{Y}^{\hbar}_{\{a\},J}\right)^{-1}\cdot\left(\mathcal{R}_{J^\bullet,\{b\}}^{\hbar}\right)^{-1} \left(\mathcal{R'}_{J^\bullet,\{b\}}^{\hbar}\right)^{-1}\mathcal{R'}_{J,J^\bullet}^{\hbar}\mathcal{R'}_{J\cup J^\bullet,\{b\}}^{\hbar} \mathcal{R'}_{\{a\},J\cup J^\bullet\cup\{b\}}^{\hbar}=
\\ \ \\
=\underbrace{R_{a,a+1}^{\hbar}\ldots R_{a,b-1}^{\hbar}}\cdot R_{b,b+1}^{\hbar}\dots R_{b,N}^{\hbar}\cdot \underbrace{R_{b-1,a}^{\hbar}\ldots R_{a+1,a}^{\hbar}} R_{a-1,a}^{\hbar}\ldots R_{1,a}^{\hbar}\cdot\mathcal{R}_{J,J^\bullet}^{\hbar}\cdot{\mathcal Y}^{\hbar}_{\mathcal N_2,J}\cdot{\mathcal Y}^{\hbar}_{\mathcal N_2,\{b\}}\cdot
\\ \ \\
\left(\mathcal{Y}^{\hbar}_{\{b\},J}\right)^{-1}\cdot\left(\mathcal{R}_{J^\bullet,\{a\}}^{\hbar}\right)^{-1} \left(\mathcal{R'}_{J^\bullet,\{a\}}^{\hbar}\right)^{-1}\mathcal{R'}_{J,J^\bullet}^{\hbar}\cdot P_{ab}\cdot \mathcal{R'}_{J\cup J^\bullet,\{b\}}^{\hbar} \mathcal{R'}_{\{a\},J\cup J^\bullet\cup\{b\}}^{\hbar}=
\\ \ \\
\displaystyle{
=\prod_{j=a+1}^{b-1} \phi(h,z_a-z_j)\phi(h,z_j-z_a)\cdot
\mathcal{R}_{\{b\},J\cup J^\bullet\cup\{a\}}^{\hbar}\mathcal{R}_{J\cup J^\bullet,\{a\}}^{\hbar}\cdot\mathcal{R}_{J,J^\bullet}^{\hbar}\cdot{\mathcal Y}^{\hbar}_{\mathcal N_2,J}\cdot{\mathcal Y}^{\hbar}_{\mathcal N_2,\{b\}}\cdot
}
\\ \ \\
\left(\mathcal{Y}^{\hbar}_{\{b\},J}\right)^{-1}\cdot\left(\mathcal{R}_{J^\bullet,\{a\}}^{\hbar}\right)^{-1} \left(\mathcal{R'}_{J^\bullet,\{a\}}^{\hbar}\right)^{-1}\mathcal{R'}_{J,J^\bullet}^{\hbar}\cdot P_{ab}\cdot \mathcal{R'}_{J\cup J^\bullet,\{b\}}^{\hbar} \mathcal{R'}_{\{a\},J\cup J^\bullet\cup\{b\}}^{\hbar}\,.
\end{array}
\eq
In the last equality we reduced the underbraced factors which give $\displaystyle \prod_{j=a+1}^{b-1} \phi(h,z_a-z_j)\phi(h,z_j-z_a)$   due to (\ref{q03}).  Also,  following the definitions (\ref{RIJ1}) and (\ref{RIJ2}) we performed transformations $R_{b,b+1}^{\hbar}\dots R_{b,N}^{\hbar}=\mathcal{R}_{\{b\},J\cup J^\bullet\cup\{a\}}^{\hbar}$  and $R_{a-1,a}^{\hbar}\ldots R_{1,a}^{\hbar}=\mathcal{R}_{J\cup J^\bullet,\{a\}}^{\hbar}$.

In a similar way we transform the r.h.s. of (\ref{e571}):
$$
\begin{array}{c}
\res\limits_{z_a=z_b}\mathcal{R}_{J\cup\{b\},J^\bullet\cup\{a\}}^{\hbar}{\mathcal Y}^{\hbar}_{\mathcal N_2,J\cup\{b\}}\mathcal{R'}_{J\cup\{b\},J^\bullet\cup\{a\}}^{\hbar}=\res\limits_{z_a=z_b}\mathcal{R}_{\{b\},J\cup J^\bullet\cup\{a\}}^{\hbar}\cdot\mathcal{R}_{J\cup J^\bullet,\{a\}}^{\hbar}\cdot\mathcal{R}_{J,J^\bullet}^{\hbar}
\\ \ \\
{\mathcal Y}^{\hbar}_{\mathcal N_2,J}\cdot{\mathcal Y}^{\hbar}_{\mathcal N_2,\{b\}}\left(\mathcal{Y}^{\hbar}_{\{b\},J}\right)^{-1}\cdot\left(\mathcal{R}_{J^\bullet,\{a\}}^{\hbar}\right)^{-1} \left(\mathcal{R'}_{J^\bullet,\{a\}}^{\hbar}\right)^{-1}\mathcal{R'}_{J,J^\bullet}^{\hbar}\cdot\mathcal{R'}_{J\cup J^\bullet,\{a\}}^{\hbar} \cdot\mathcal{R'}_{\{b\},J\cup J^\bullet\cup\{a\}}^{\hbar}\,.=
\\ \ \\
=-\mathcal{R}_{\{b\},J\cup J^\bullet\cup\{a\}}^{\hbar}\cdot\mathcal{R}_{J\cup J^\bullet,\{a\}}^{\hbar}\cdot\mathcal{R}_{J,J^\bullet}^{\hbar}\cdot
{\mathcal Y}^{\hbar}_{\mathcal N_2,J}\cdot{\mathcal Y}^{\hbar}_{\mathcal N_2,\{b\}}\left(\mathcal{Y}^{\hbar}_{\{b\},J}\right)^{-1}\cdot\left(\mathcal{R}_{J^\bullet,\{a\}}^{\hbar}\right)^{-1} \left(\mathcal{R'}_{J^\bullet,\{a\}}^{\hbar}\right)^{-1}\mathcal{R'}_{J,J^\bullet}^{\hbar}\cdot
\end{array}
$$
\beq\label{e578}
\begin{array}{c}
R_{N,a}^{\hbar}\dots R_{b+1,a}^{\hbar}R_{b-1,a}^{\hbar}\dots R_{a+1,a}^{\hbar}\cdot R_{b,1}^{\hbar}\dots  R_{b,a-1}^{\hbar}\cdot P_{ab}\cdot R_{b,a+1}^{\hbar}\dots  R_{b,b-1}^{\hbar}
\end{array}
\eq
and move permutation $P_{ab}$ to the left:
\beq\label{e579}
\begin{array}{c}
\res\limits_{z_a=z_b}\mathcal{R}_{J\cup\{b\},J^\bullet\cup\{a\}}^{\hbar}{\mathcal Y}^{\hbar}_{\mathcal N_2,J\cup\{b\}}\mathcal{R'}_{J\cup\{b\},J^\bullet\cup\{a\}}^{\hbar} =-\mathcal{R}_{\{b\},J\cup J^\bullet\cup\{a\}}^{\hbar}\cdot\mathcal{R}_{J\cup J^\bullet,\{a\}}^{\hbar}\cdot\mathcal{R}_{J,J^\bullet}^{\hbar}
\\ \ \\
{\mathcal Y}^{\hbar}_{\mathcal N_2,J}\cdot{\mathcal Y}^{\hbar}_{\mathcal N_2,\{b\}}\left(\mathcal{Y}^{\hbar}_{\{b\},J}\right)^{-1}\cdot\left(\mathcal{R}_{J^\bullet,\{a\}}^{\hbar}\right)^{-1} \left(\mathcal{R'}_{J^\bullet,\{a\}}^{\hbar}\right)^{-1}\mathcal{R'}_{J,J^\bullet}^{\hbar}\cdot P_{ab}
\\ \ \\
 R_{N,b}^{\hbar}\dots R_{b+1,b}^{\hbar}\underbrace{R_{b-1,b}^{\hbar}\dots R_{a+1,b}^{\hbar}}\cdot R_{a,1}^{\hbar}\dots  R_{a,a-1}^{\hbar}\cdot \underbrace{R_{b,a+1}^{\hbar}\dots  R_{b,b-1}^{\hbar}}=
 \\ \ \\
 -\mathcal{R}_{\{b\},J\cup J^\bullet\cup\{a\}}^{\hbar}\cdot\mathcal{R}_{J\cup J^\bullet,\{a\}}^{\hbar}\cdot\mathcal{R}_{J,J^\bullet}^{\hbar}\cdot{\mathcal Y}^{\hbar}_{\mathcal N_2,J}\cdot{\mathcal Y}^{\hbar}_{\mathcal N_2,\{b\}}\left(\mathcal{Y}^{\hbar}_{\{b\},J}\right)^{-1}\cdot\left(\mathcal{R}_{J^\bullet,\{a\}}^{\hbar}\right)^{-1} \left(\mathcal{R'}_{J^\bullet,\{a\}}^{\hbar}\right)^{-1}
 \\ \  \\
 {\displaystyle \mathcal{R'}_{J,J^\bullet}^{\hbar}\cdot P_{ab} \cdot\mathcal{R'}_{J\cup J^\bullet,\{b\}}^{\hbar} \mathcal{R'}_{\{a\},J\cup J^\bullet\cup\{b\}}^{\hbar}\prod_{j=a+1}^{b-1} \phi(h,z_a-z_j)\phi(h,z_j-z_a)\,.
 }
\end{array}
\eq
In the last equality we reduced the underbraced factors which give $\displaystyle \prod_{j=a+1}^{b-1} \phi(h,z_a-z_j)\phi(h,z_j-z_a)$   due to (\ref{q03}) and used the definitions (\ref{RIJ'1}) and (\ref{RIJ'2}). The r.h.s. of (\ref{e579}) equals minus r.h.s. of  (\ref{e577}). This finishes the proof.
\end{proof}

\subsection*{Proof of Lemma \ref{Lem32}}

\begin{proof}
The proof is performed for $\mF_1$ and $\mF_2$ separately, that is the relation (\ref{e58}) holds true for $\mF_1$ and $\mF_2$ separately. The poles at $z_a=z_{b}$ are contained in the middle products $\mathcal{Y}_{\mathcal N_2,I}^{\hbar}$ and $\mathcal{Y}_{\mathcal N_1,J}^{-\hbar}$ in
(\ref{e531})-(\ref{e532}) for such $I$ and $J$ that $a\in I$ and $b\in J$.

The calculation for $\mF_1$ is straightforward. For any set $I$ which contains $a$  the product $\mathcal{Y}_{\mathcal N_2,I}^{\hbar}$   has only one item with $R_{ba}^{\hbar}$ and thus has  simple pole at $z_a=z_{b}$. Due to the property (\ref{r05}) one should replace $R_{ba}$ with the permutation $-P_{ab}$ and put down  $z_a=z_{b}$.
To transform the residue to the form (\ref{e58}) we move the permutation operator to the right.

Firstly, we rewrite each item in $\mathcal F_1$ separating $R$-matrices containing indices $a$ and $b$. For this purpose relations (\ref{l21})-(\ref{l22}) and (\ref{lemmaY1})-(\ref{lemmaY2}) are used. Let us fix $$I=\{i_1<\dots<i_{m-1}<i_m=a<i_{m+1}<\dots <i_k\}$$ and transform $\mathcal{R}_{I,I^c}^{h}$ using (\ref{l21}):

\begin{equation}\label{e601}
\begin{array}{lll}
  \mathcal{R}_{I,I^c}^{h}=\mathcal{R}_{\{a\},\mathcal N_1\setminus\{a\}}^{h}\cdot \mathcal{R}_{I\setminus\{a\},I^c}^{h} \cdot\left(\mathcal{R}_{\{a\},I\setminus \{a\}}^{h}\right)^{-1}=
  \\ \ \\
  =R_{a,a+1}^{\hbar}\dots R_{a,N}^{\hbar} \cdot \mathcal{R}_{I\setminus\{a\},I^c}^{h}\cdot \left(R_{a,i_{m+1}}^{\hbar}\dots R_{a,i_{k}}^{\hbar} \right)^{-1}\,.
\end{array}
  \end{equation}
Similar transformations are performed using (\ref{l22}):
\begin{equation}\label{e602}
\begin{array}{lll}
   \mathcal{R'}_{I,I^c}^{\hbar}=\left(\mathcal{R'}_{\{a\},I\setminus \{a\}}^{h}\right)^{-1}\cdot \mathcal{R'}_{I\setminus\{a\},I^c}^{h} \cdot\mathcal{R'}_{\{a\},\mathcal N_1\setminus\{a\}}^{h}=
  \\ \ \\
   =\left(R_{a,i_1}^{\hbar}\dots R_{a,i_{m-1}}^{\hbar}\right)^{-1} \cdot \mathcal{R'}_{I\setminus\{a\},I^c}^{h}\cdot R_{a,1}^{\hbar}\dots R_{a,a-1}^{\hbar}\,.
\end{array}
  \end{equation}
  By applying consistently  (\ref{lemmaY1})-(\ref{lemmaY2}) to $\mathcal{Y}_{\mathcal N_2,I}^{\hbar}$ the latter is rewritten as follows:
  \begin{equation}\label{e603}
\begin{array}{lll}
  {\mathcal Y}^{\hbar}_{\mathcal N_2,I}=R^{\hbar}_{a,i_{m+1}}\dots R^{\hbar}_{a,i_k}\cdot R^{\hbar}_{b+1,b}\dots R^{\hbar}_{2N,b}\cdot{\mathcal Y}^{\hbar}_{\mathcal N_2\setminus\{b\},I\setminus \{a\}}\cdot{\mathcal Y}^{\hbar}_{\{b\},I\setminus \{a\}}\cdot{\mathcal Y}^{\hbar}_{\mathcal N_2\setminus\{b\},\{a\}}\cdot
  \\ \ \\
  \cdot R^{\hbar}_{ba}\cdot\left(R^{\hbar}_{a,i_{m+1}}\dots R^{\hbar}_{a,i_k}\cdot R^{\hbar}_{b+1,b}\dots R^{\hbar}_{2N,b}\right)^{-1}\,.
\end{array}
  \end{equation}
Let us calculate the residue of $\mF_1$ at $z_a=z_{b}$. The nontrivial part is  ${\mathcal Y}^{\hbar}_{\mathcal N_2,I}$, while the items $\mathcal{R}_{I,I^c}^{\hbar}$ and $\mathcal{R'}_{I,I^c}^{\hbar}$ has no pole. Plugging $${\mathcal Y}^{\hbar}_{\{b\},I\setminus \{a\}}=R^{\hbar}_{b,i_1}\dots  R^{\hbar}_{b,i_{m-1}}R^{\hbar}_{b,i_{m+1}}\dots R^{\hbar}_{b,i_{k}}$$ and $${\mathcal Y}^{\hbar}_{\mathcal N_2\setminus\{b\},\{a\}}= R^{\hbar}_{N+1,a}\dots R^{\hbar}_{b-1,a}R^{\hbar}_{b+1,a}\dots R^{\hbar}_{2N,a}$$ into (\ref{e603})
 one obtains:
  \begin{equation}\label{e604}
\begin{array}{lll}
  \res\limits_{z_a=z_b}{\mathcal Y}^{\hbar}_{\mathcal N_2,I}=
  -R^{\hbar}_{a,i_{m+1}}\dots R^{\hbar}_{a,i_k}\cdot R^{\hbar}_{b+1,b}\dots R^{\hbar}_{2N,b}\cdot{\mathcal Y}^{\hbar}_{\mathcal N_2\setminus\{b\},I\setminus \{a\}}\cdot
  \\ \ \\
  R^{\hbar}_{b,i_1}\dots  R^{\hbar}_{b,i_{m-1}}R^{\hbar}_{b,i_{m+1}}\dots R^{\hbar}_{b,i_{k}}\cdot R^{\hbar}_{N+1,a}\dots R^{\hbar}_{b-1,a}R^{\hbar}_{b+1,a}\dots R^{\hbar}_{2N,a}\cdot
  \\ \ \\
  P_{ab}\left(R^{\hbar}_{a,i_{m+1}}\dots R^{\hbar}_{a,i_k}\cdot R^{\hbar}_{b+1,b}\dots R^{\hbar}_{2N,b}\right)^{-1}\,.
\end{array}
  \end{equation}
  In order to simplify we move the permutation operator $P_{ab}$ to the right and then reduce the underlined items:
   \begin{equation}\label{e605}
\begin{array}{lll}
  \res\limits_{z_a=z_b}{\mathcal Y}^{\hbar}_{\mathcal N_2,I}=
  -R^{\hbar}_{a,i_{m+1}}\dots R^{\hbar}_{a,i_k}\cdot R^{\hbar}_{b+1,b}\dots R^{\hbar}_{2N,b}\cdot{\mathcal Y}^{\hbar}_{\mathcal N_2\setminus\{b\},I\setminus \{a\}}\cdot
  \\ \ \\
  R^{\hbar}_{b,i_1}\dots  R^{\hbar}_{b,i_{m-1}}\underline{R^{\hbar}_{b,i_{m+1}}\dots R^{\hbar}_{b,i_{k}}}\cdot R^{\hbar}_{N+1,a}\dots R^{\hbar}_{b-1,a}\underline{R^{\hbar}_{b+1,a}\dots R^{\hbar}_{2N,a}}\cdot
  \\ \ \\
 \left(\underline{R^{\hbar}_{b,i_{m+1}}\dots R^{\hbar}_{b,i_k}\cdot R^{\hbar}_{b+1,a}\dots R^{\hbar}_{2N,a}}\right)^{-1} P_{ab}=
 \\ \ \\
 -R^{\hbar}_{a,i_{m+1}}\dots R^{\hbar}_{a,i_k}\cdot R^{\hbar}_{b+1,b}\dots R^{\hbar}_{2N,b}\cdot{\mathcal Y}^{\hbar}_{\mathcal N_2\setminus\{b\},I\setminus \{a\}}\cdot R^{\hbar}_{b,i_1}\dots  R^{\hbar}_{b,i_{m-1}}\cdot R^{\hbar}_{N+1,a}\dots R^{\hbar}_{b-1,a}\cdot P_{ab}\,.
\end{array}
  \end{equation}
 Indeed, expression $R^{\hbar}_{b,i_{m+1}}\dots R^{\hbar}_{b,i_{k}}$ commutes with $R^{\hbar}_{N+1,a}\dots R^{\hbar}_{b-1,a}$ due to (\ref{QYB3}) since each multiplier has pairwise different indices. That is the underlined terms are cancelled.

 Now we are going to calculate the residue at $z_a=z_b$ of each summand containing index $a\in I$  in $\mF_1$ given in (\ref{e531}). We use (\ref{e601}), (\ref{e602}) and (\ref{e605}):
   \begin{equation}\label{e606}
  \begin{array}{lll}
  \res\limits_{z_a=z_b} \mathcal{R}_{I,I^c}^{\hbar}\cdot{\mathcal Y}^{\hbar}_{\mathcal N_2,I} \cdot\mathcal{R'}_{I,I^c}^{\hbar}=-R_{a,a+1}^{\hbar}\dots R_{a,N}^{\hbar} \cdot \mathcal{R}_{I\setminus\{a\},I^c}^{h}\cdot \underbrace{R^{\hbar}_{b+1,b}\dots R^{\hbar}_{2N,b}}\cdot{\mathcal Y}^{\hbar}_{\mathcal N_2\setminus\{b\},I\setminus \{a\}}\cdot
  \\ \ \\
 R^{\hbar}_{b,i_1}\dots  R^{\hbar}_{b,i_{m-1}}\cdot R^{\hbar}_{N+1,a}\dots R^{\hbar}_{b-1,a}\cdot P_{ab}\cdot\left(R_{a,i_1}^{\hbar}\dots R_{a,i_{m-1}}^{\hbar}\right)^{-1} \cdot \mathcal{R'}_{I\setminus\{a\},I^c}^{h}\cdot R_{a,1}^{\hbar}\dots R_{a,a-1}^{\hbar}\,.
\end{array}
\end{equation}
Next, we move $P_{ab}$ to the right and the underbraced item to the left through  $\mathcal{R}_{I\setminus\{a\},I^c}^{h}$ due to (\ref{r08}) since it does not contain indices $b, b+1,\dots 2N$:
   \begin{equation}\label{e607}
  \begin{array}{lll}
 \res\limits_{z_a=z_b} \mathcal{R}_{I,I^c}^{\hbar}\cdot{\mathcal Y}^{\hbar}_{\mathcal N_2,I} \cdot\mathcal{R'}_{I,I^c}^{\hbar}=-R_{a,a+1}^{\hbar}\dots R_{a,N}^{\hbar} \cdot  R^{\hbar}_{b+1,b}\dots R^{\hbar}_{2N,b}\cdot\mathcal{R}_{I\setminus\{a\},I^c}^{h}\cdot{\mathcal Y}^{\hbar}_{\mathcal N_2\setminus\{b\},I\setminus \{a\}}\cdot
\\ \ \\
 \underline{R^{\hbar}_{b,i_1}\dots  R^{\hbar}_{b,i_{m-1}}}\cdot R^{\hbar}_{N+1,a}\dots R^{\hbar}_{b-1,a}\cdot \left(\underline{R_{b,i_1}^{\hbar}\dots R_{b,i_{m-1}}^{\hbar}}\right)^{-1} \cdot \mathcal{R'}_{I\setminus\{a\},I^c}^{h}\cdot R_{b,1}^{\hbar}\dots R_{b,a-1}^{\hbar}\cdot P_{ab}=
 \\ \ \\
 =-G\cdot \mathcal{R}_{I\setminus\{a\},I^c}^{h}\cdot  {\mathcal Y}^{\hbar}_{\mathcal N_2\setminus\{b\},I\setminus \{a\}}\cdot R^{\hbar}_{N+1,a}\dots R^{\hbar}_{b-1,a}\cdot
 \mathcal{R'}_{I\setminus\{a\},I^c}^{h}\cdot
R_{b,1}^{\hbar}\dots R_{b,a-1}^{\hbar}\cdot P_{ab}\,.
 \end{array}
 \end{equation}
 In the last equality we combined the first terms into $G$ using the definition (\ref{e59}) and reduced the underlined items. Then  moving $R^{\hbar}_{N+1,a}\dots R^{\hbar}_{b-1,a}$ through $\mathcal{R'}_{I\setminus\{a\},I^c}^{h}$ to the right and using the definition  (\ref{e60}) one obtains:
 \begin{equation}\label{e608}
  \begin{array}{lll}
 \res\limits_{z_a=z_b} \mathcal{R}_{I,I^c}^{\hbar}\cdot{\mathcal Y}^{\hbar}_{\mathcal N_2,I} \cdot\mathcal{R'}_{I,I^c}^{\hbar}=-G\cdot \mathcal{R}_{I\setminus\{a\},I^c}^{h}\cdot  {\mathcal Y}^{\hbar}_{\mathcal N_2\setminus\{b\},I\setminus \{a\}}\cdot
 \mathcal{R'}_{I\setminus\{a\},I^c}^{h}\cdot H \,.
\end{array}
  \end{equation}
  This yields
     \beq\label{e6081}
  \begin{array}{c}
  \displaystyle{
  \res\limits_{z_a=z_{b}}\mF_1[k,N]=-G\cdot \mF_1[k-1,N-1]\cdot H\,.
 }
 \end{array}
 \eq
 The same calculation can be performed for $\mF_2$ since it  differs from $\mF_1$ by only changing $I$ to $J$, $\mathcal N_2$ to $\mathcal N_1$, $\hbar\to-\hbar$ and the total sign.  Therefore, we use the result of (\ref{e608}):
 \begin{equation}\label{e609}
  \res\limits_{z_a=z_b}\mathcal{R}_{J,J^c}^{-\hbar} {\mathcal Y}^{-\hbar}_{\mathcal N_1,J} \mathcal{  R'}_{J,J^c}^{-\hbar}=-G'\cdot \mathcal{R}_{J\setminus\{b\},J^c}^{-h}\cdot  {\mathcal Y}^{-\hbar}_{\mathcal N_1\setminus\{a\},J\setminus \{b\}}\cdot
 \mathcal{R'}_{J\setminus\{b\},J^c}^{-h}\cdot H' ,
  \end{equation}
  where
  \beq\label{e610}
  G'=R^{-\hbar}_{b,b+1}\ldots R^{-\hbar}_{b,2N} \cdot R^{-\hbar}_{a+1,a}\ldots R^{-\hbar}_{N,a}\,,
  \eq

  \beq\label{e611}
H'=R^{-\hbar}_{1,b}\ldots R_{a-1,b}^{-\hbar}\cdot R^{-\hbar}_{a,N+1}\ldots R^{-\hbar}_{a,b-1}P_{ab}\,.
  \eq
We change each $R_{i,j}^{-\hbar}$ to $-R_{j,i}^{\hbar}$ in (\ref{e610}) and (\ref{e611})  using (\ref{r08}):
 \beq\label{e612}
  G'=(-1)^{3N-a-b}R^{\hbar}_{b+1,b}\ldots R^{\hbar}_{2N,b} \cdot R^{\hbar}_{a,a+1}\ldots R^{\hbar}_{a,N}=(-1)^{a+b-N}G\,,
  \eq
    \beq\label{e613}
H'=(-1)^{a-1+b-N-1}R^{\hbar}_{b,1}\ldots R_{b,a-1}^{\hbar}\cdot R^{\hbar}_{N+1,a}\ldots R^{\hbar}_{b-1,a}P_{ab}=(-1)^{a+b-N-1}H\,.
  \eq
  Due to (\ref{e612})-(\ref{e613}) one can change $G'$ in (\ref{e609}) to $G$ and $H'$ to $H$ thus providing
   \beq\label{e614}
  \begin{array}{c}
  \displaystyle{
  \res\limits_{z_a=z_{b}}\mF_2[k,N]=-G\cdot \mF_2[k-1,N-1]\cdot H\,.
 }
 \end{array}
 \eq
  This finishes the proof.
\end{proof}


\subsection*{Acknowledgments}


We are grateful to A. Levin, A. Liashyk and S. Khoroshkin for useful discussions.
The work of M. Matushko was partially supported by Russian Science Foundation (project 20-41-09009).



\begin{small}

\end{small}

\end{document}